\documentclass[a4paper,10pt]{article}
\usepackage[bbgreekl]{mathbbol}
\usepackage[small]{titlesec}
\usepackage[utf8]{inputenc}
\usepackage[T1]{fontenc}  
\usepackage[]{mdframed}
\usepackage{mathrsfs}

\usepackage[backend=biber,style=alphabetic,natbib=false,giveninits=true,doi=true,isbn=false,url=false,date=year,maxbibnames=99,sorting=nty,defernumbers=true]{biblatex}
\DeclareNameAlias{default}{family-given}

\DeclareFieldFormat[article]{title}{#1}
\renewbibmacro{in:}{}
\DeclareFieldFormat[article]{pages}{#1}
\DeclareFieldFormat{journal}{#1}
\renewcommand\bf\bfseries

\renewbibmacro*{volume+number+eid}{%
	\printfield{volume}%
	\setunit*{\addnbspace}
	\printfield{number}%
	\setunit{\addcomma\space}%
	\printfield{eid}}
\DeclareFieldFormat[article]{number}{\mkbibparens{#1}}
\DeclareFieldFormat[article]{volume}{\textbf{#1}}
\DeclareFieldFormat{year}{\mkbibparens{#1}}
\DeclareBibliographyDriver{article}{%
	\printnames{author}:%
	\newunit\newblock
	\printfield{title}%
	\newunit\newblock
	\printfield{journaltitle}
	\newunit
	\iffieldundef{number}{\printfield{volume}}{\printfield{volume}\addspace\printfield{number}}
	\addcomma\addspace\printfield{pages}\addspace
	\printfield{year}
}
\AtEveryBibitem{\clearfield{month}}
\AtEveryBibitem{\clearfield{day}}
\usepackage[english]{babel}

\usepackage{csquotes}
\usepackage{bbm}
\usepackage{amsmath}
\usepackage{amsfonts,amsthm,amsbsy,amssymb,dsfont,stmaryrd}
\usepackage[dvipsnames]{xcolor}
\usepackage{braket}

\usepackage{caption}
\usepackage{subcaption}

\usepackage{enumitem}

\makeatletter
\newcommand{\leqnomode}{\tagsleft@true\let\veqno\@@leqno}
\newcommand{\reqnomode}{\tagsleft@false\let\veqno\@@eqno}
\makeatother

\usepackage{mathtools}
\usepackage[makeroom]{cancel}

\numberwithin{equation}{section}

\newcommand\myshade{85}
\colorlet{mylinkcolor}{violet}
\colorlet{mycitecolor}{YellowOrange}
\colorlet{myurlcolor}{Aquamarine}

\usepackage[left=2.5cm,right=2.5cm,top=2.5cm,bottom=2.5cm]{geometry}
\usepackage[unicode=true,pdfusetitle,
bookmarks=true,bookmarksnumbered=false,bookmarksopen=false,
breaklinks=false,pdfborder={0 0 1},backref=false,
linkcolor  = mylinkcolor!\myshade!black,
citecolor  = mycitecolor!\myshade!black,
urlcolor   = myurlcolor!\myshade!black,
colorlinks = true,
]
{hyperref}
\usepackage[nameinlink]{cleveref}

\usepackage{tikz}
\usetikzlibrary{hobby}
\usetikzlibrary{arrows}
\usetikzlibrary{fit,backgrounds}
\usetikzlibrary{decorations.pathreplacing,decorations.markings}
\tikzset{arrow data/.style 2 args={%
      decoration={%
         markings,
         mark=at position #1 with \arrow{#2}},
         postaction=decorate}
      }%

\usepackage{tikz-cd}
\usepackage{ifthen}

\usepackage{graphicx}

\usepackage{slashed}
\definecolor{ct_black}{HTML}{000000}
\definecolor{ct_orange}{HTML}{ED872D}
\definecolor{ct_purple}{HTML}{7A68A6}
\definecolor{ct_blue}{HTML}{348ABD}
\definecolor{ct_turquoise}{HTML}{188487}
\definecolor{ct_red}{HTML}{E32636}
\definecolor{ct_pink}{HTML}{CF4457}
\definecolor{ct_green}{HTML}{467821}

\definecolor{ct2_green}{HTML}{9FF781}
\definecolor{ct2_green_dark}{HTML}{088A08}

\theoremstyle{plain}
\newtheorem{thm}{\protect\theoremname}[section]
\theoremstyle{plain}
\newtheorem{lem}[thm]{\protect\lemmaname}
\theoremstyle{plain}
\newtheorem{cor}[thm]{\protect\corollaryname}
\theoremstyle{plain}

\newtheorem{proposition}[thm]{\protect\propositionname}
\theoremstyle{plain}
\newtheorem{claim}[thm]{\protect\claimname}

\theoremstyle{remark}
\newtheorem{rem}[thm]{\protect\remarkname}

\theoremstyle{definition}
\newtheorem{defn}[thm]{\protect\definitionname}

\theoremstyle{plain}

\providecommand{\assumptionname}{Assumption}
\providecommand{\claimname}{Claim}
\providecommand{\corollaryname}{Corollary}
\providecommand{\definitionname}{Definition}
\providecommand{\lemmaname}{Lemma}
\providecommand{\propositionname}{Proposition}
\providecommand{\remarkname}{Remark}
\providecommand{\theoremname}{Theorem}
\providecommand{\examplename}{Example}

\crefname{section}{Section}{Sections}
\crefname{example}{Example}{Examples}
\crefname{appendix}{Appendix}{Appendices}
\crefname{figure}{Figure}{Figures}
\crefname{assumption}{Assumption}{Assumptions}
\crefname{thm}{Theorem}{Theorems}
\crefname{lem}{Lemma}{Lemmas}
\crefformat{equation}{(#2#1#3)}
\crefname{table}{Table}{Tables}

\crefrangelabelformat{equation}{(#3#1#4--#5#2#6)}

\crefmultiformat{equation}{(#2#1#3}{, #2#1#3)}{#2#1#3}{#2#1#3}
\Crefmultiformat{equation}{(#2#1#3}{, #2#1#3)}{#2#1#3}{#2#1#3}

\newtheorem*{lem*}{\protect\lemmaname}

\newcommand{\ee}{\operatorname{e}}
\newcommand{\ii}{\operatorname{i}}

\newcommand{\ZZ}{\mathbb{Z}}

\newcommand{\NN}{\mathbb{N}}
\newcommand{\RR}{\mathbb{R}}

\newcommand{\CC}{\mathbb{C}}

\newcommand{\calB}{\mathcal{B}}

\newcommand{\calH}{\mathcal{H}}

\newcommand{\bbLambda}{\mathbb{\Lambda}}

\newcommand\norm[1]{\left\lVert#1\right\rVert}
\newcommand\abs[1]{\left|#1\right|}
\newcommand{\ip}[2]{\langle #1, #2 \rangle}

\newcommand{\dif}{\operatorname{d}\!} 
\newcommand{\tr}{\operatorname{tr}}

\newcommand{\Id}{\mathds{1}}

\newcommand{\dist}{\mathrm{dist}}

\newcommand{\Chern}{\operatorname{Chern}}
\newcommand{\Zak}{\operatorname{Zak}}

\newcommand{\sgn}{\operatorname{sgn}}
\newcommand{\findex}{\operatorname{index}}
\newcommand{\ind}{\findex}

\newcommand{\supp}{\operatorname{supp}}
\newcommand{\Sgntr}{\operatorname{Signature}}
\newcommand{\Sf}{\operatorname{Sf}}

\newcommand{\br}[1]{\left(#1\right)}

\DeclareMathOperator{\spec}{\sigma}

\newcommand{\im}{\operatorname{im}}

\newcommand{\gap}{\operatorname{gap}}

\usepackage{environ}

\NewEnviron{malign}{%
	\begin{align}\begin{split}
			\BODY
	\end{split}\end{align}
}

\newcommand{\eq}[1]{\begin{align*}#1\end{align*}}
\newcommand{\eql}[1]{\begin{align}#1\end{align}}

\newcommand{\polar}{\operatorname{pol}}

\title{The Loring--Schulz-Baldes Spectral Localizer Revisited}
\author{\href{mailto:gberkolaiko@tamu.edu}{Gregory Berkolaiko}\\
	{\footnotesize Department of Mathematics, Texas A\&M University}\\
    \href{mailto:jacobshapiro@princeton.edu}{Jacob Shapiro}\\
	{\footnotesize Department of Mathematics, Princeton University}\\
 \href{mailto:beyer.white@gmail.com}{Beyer Chase White}\\
	{\footnotesize IST Austria }
}

\begin{document}
\reqnomode

\maketitle

\begin{abstract}
  The spectral localizer, introduced by Loring in 2015 and
  Loring and Schulz-Baldes in 2017, is a method to compute the
  (infinite volume) topological invariant of a quantum Hamiltonian on
  $\ZZ^d$, as the signature of the (finite) localizer matrix.
  We present a direct and elementary spectral-theoretic proof treating the $d=1$ and $d=2$ cases on an
  almost equal footing. Moreover, we re-interpret the localizer as a higher-dimensional topological insulator via the bulk-edge correspondence.
\end{abstract}

\section{Introduction}
In the study of topological insulators \cite{Hasan_Kane_2010}, quantum mechanical Hamiltonians (or their associated Fermi projections) are labelled by a topological index: a continuous map from the space of such systems into a discrete group. The most well-known example being the two-dimensional Chern number, which also corresponds to an experimentally measurable quantity, the Hall conductivity. Mathematically, these indices arise in various forms: as winding numbers of functions, as Chern characteristic classes of vector bundles, as index-pairings between K-theory group elements and a Dirac operator, or as Fredholm indices of certain operators on a Hilbert space. With no exception, all of these topological indices are zero if the system is truncated to a finite volume in real space. Physically, the intuition for this is provided by a finite-width Hall insulator: it will have edge currents running in opposite directions on its two opposite boundaries, producing zero in total. Mathematically, one perspective on this problem is afforded by the Fredholm index being zero for any finite square matrix: by the rank-nullity theorem, the dimension of the kernel equals the dimension of the cokernel.  Another way to view this issue is to note that all topological indices may be written in the form $\tr([A,B])$ for appropriate operators $A,B$ (for instance, the Chern number may be written as $2\pi\ii\tr([P\Lambda_1P,P\Lambda_2P])$ with $P$ the Fermi projection and $\Lambda_j\equiv\chi_{\NN}(X_j)$, $j=1,2$). If these are truncated to finite space so that $A,B$ are matrices, then $\tr(AB) = \tr(BA)$ and the expression vanishes.

The issue just described poses a problem when trying to run numerical calculations on topological insulators. To be sure, most toy models in physics are \emph{translation invariant}, in which case the indices reduce to winding of vector bundles, which can be evaluated numerically despite the system being infinite volume. However, for \emph{disordered} systems, it is not really clear how to proceed. To that end, in pioneering studies \cite{Lor15,LSB1}, Loring, and later, Loring and Schulz-Baldes introduced a new way to run finite-volume calculations: \emph{the spectral localizer}. In this context, one should also mention \cite{Prodan2017,Taub18,Hastings2010} which offer alternative perspectives on this question. 

The spectral localizer construction is as follows: 
\begin{itemize}
\item Add to the Hamiltonian an \emph{unbounded} self-adjoint perturbation which however keeps the gap open due to its special structure.
\item Truncate the result to a finite Hermitian matrix.  The truncation is performed where the perturbation is large; since the low energy eigenfunctions are exponentially suppressed in area of large perturbation, truncation has a vanishingly small effect on the gap states.
\item Compute the signature (the number of positive minus the number of negative eigenvalues) of the truncated matrix. 
\end{itemize} 
The result is twice the topological index of the infinite-volume Hamiltonian one had started with. The first proof used K-theory of fuzzy spheres (for $d=1$ using the eta invariant) \cite{LSB1, LSB2} and the second proof \cite{EJSB, LSB3} used Phillips' theorem connecting index pairings and the spectral flow \cite{Phillips}.

Our goal here is to re-derive the Loring--Schulz-Baldes result, motivate it via the bulk-edge correspondence and provide an elementary spectral flow-based argument that treats both odd and even dimensions in a unified way (cf. \cite{LoringSchulzBaldes2019SpectralFlowOdd,LozanoViescaSchoberSchulzBaldes2019ChernNumbers} which separate into even and odd $d$). Moreover, our choice of the perturbation is slightly different from the original one, leading to more convenient estimates and a bounded infinite-volume localizer.

\subsection{Setting}
We consider non-interacting insulators in $d$ space
dimensions. Hence the appropriate Hilbert space is $\calH :=
\ell^2(\ZZ^d)\otimes\CC^{N}$ for some (fixed once and for all)
$N\in\NN$, the internal number of degrees of freedom. We are
interested in \emph{insulators} which in this paper means a
spectral gap at $E_F:=0$ (without loss of generality). Moreover,
we assume locality, 
\eql{\label{eq:locality}\norm{H_{x,y}}\leq
  C\ee^{-\mu\norm{x-y}}\qquad(x,y\in\ZZ^d)} for some
$C<\infty,\mu>0$. Here by $H_{x,y}$ we mean the $N\times N$
matrix whose $i,j$ matrix elements is given by
$\ip{\delta_x\otimes e_i}{H (\delta_y\otimes e_j)}$ with
$\Set{e_j}_j$ the standard basis of $\CC^N$ and $\Set{\delta_x}_{x\in\ZZ^d}$ the
canonical position basis on $\ell^2(\ZZ^d)$.

Let us specify to $d=1$ for simplicity. When present, the chiral symmetry condition is specified as follows: if $\calH_{\pm}:=\ell^2(\ZZ)\otimes\CC^{N/2}$ is half the Hilbert space (for chiral systems we assume $N\in2\NN$) so that $\calH \cong \calH_+\oplus\calH_-$ then in that direct sum decomposition, $H$ is given as an off-diagonal block operator of the form \eql{\label{eq:chiral constraint}H = \begin{bmatrix}
    0 & S^\ast \\
    S & 0
  \end{bmatrix} } for some $S\in\calB(\calH_+\to\calH_-)$. The first two properties of $H$ are inherited by the operator $S$: it is invertible and obeys a similar estimate as \Cref{eq:locality}.

In this setting it is well-known that in the chiral $d=1$ case \cite{PSB_2016,Graf_Shapiro_2018_1D_Chiral_BEC}, associated to $H$ is a topological index, the Zak phase \cite{Zak_PhysRevLett.62.2747} \footnote{Originally defined only in the translation invariant setting.}. It is described as follows: let $\Lambda := \chi_{\NN}(X)$ be the projection onto the right-half space ($\chi$ is the characteristic function and $X$ is the position operator on $\calH$). Then the Zak phase is given by the Fredholm index
\begin{align} \label{eq:Zak phase_S} \Zak(H) &:= \findex(\Lambda S \Lambda + \Lambda^\perp)\\ \label{eq:Zak phase_U} &= \findex(\Lambda S^\flat \Lambda + \Lambda^\perp)\in\ZZ\, \end{align}
where $S^\flat:=\operatorname{pol}(S)$ is the polar part of the operator $S$.
To see that the operators in \Cref{eq:Zak phase_S,eq:Zak phase_U} are Fredholm, we note that  \Cref{eq:locality} implies that $[\Lambda,S]$ is compact which is a property inherited by $[\Lambda,S^\flat]$ (as long as the gap is open). This in turn implies the Fredholm property by Atkinson's theorem.

So far we have been discussing what is known as the \emph{spectral gap regime} (where $H$ has a spectral gap).  We note in passing that in the mobility gap regime (when $H$ has no spectral gap but the states about $E_F=0$ obey a dynamical estimate which still makes $H$ into an insulator; see \cite{EGS_2005,Shapiro20}) the appropriate definition is only via \Cref{eq:Zak phase_U} and not \Cref{eq:Zak phase_S}, since the operator appearing on the right-hand side of \Cref{eq:Zak phase_S} does not yield a Fredholm operator. Finally, we (re-)emphasize that if we were to truncate $H$ to a finite volume and attempt the calculation of the index the result would always be zero.

In the non-chiral $d=2$ case we have the integer quantum Hall
effect, whose associated physical observable is the Hall
conductivity, connected to the Chern number \cite{Graf07}. We
write it, following \cite{Bellissard_1994JMP....35.5373B}, as
\eql{ \Chern(H) := \findex(PLP+P^\perp)\in\ZZ } with
$P:=\chi_{(-\infty,0)}(H)$ the Fermi projection and
$L:=\exp\br{\ii\arg\br{X_1+\ii X_2}}$ the Laughlin flux insertion
operator. It is well known (see e.g. \cite[Lemma A.1]{BSS23}) that
\Cref{eq:locality} implies $[P,L]$ is compact, which then implies
the operator in the index is Fredholm again by Atkinson's theorem.

\subsection{The Loring--Schulz-Baldes spectral localizer construction}

We describe now the spectral localizer construction and the theorem that relates
the signature of its finite volume truncation to the Zak phase
$\Zak(H)$ and to the Chern number $\Chern(H)$.

\newcommand{\Bell}{{\mathcal{B}_\ell}}
For any $\ell\in\NN$, let \eq{
\calB_\ell := \Set{x \in \ZZ^d | \norm{x} \leq \ell}
} be the box of radius $\ell$
and 
\begin{equation}
  \label{eq:Hell_def}
  \calH_\ell := \CC^{\Bell} \otimes \CC^N  
\end{equation}
be the finite-volume Hilbert space based on $\Bell$.  Viewing
$\calH_\ell$ as a subspace of $\calH$, let \eq{J_\ell:
  \calH_{\ell}\to\calH} be the inclusion operator (i.e., $J_\ell$
extends a wave-function by zero outside the finite box). Then
$J_\ell^* J_\ell:\calH_\ell \to \calH_\ell$ is the identity and
$J_\ell J_\ell^\ast:\calH\to\calH$ is the orthogonal projector onto
$\calH_\ell$.  With this notation,
\eql{\label{eq:H_truncation}H_{\ell}:=J_{\ell}^{\ast}H J_{\ell}} is
the finite-volume (Dirichlet) truncation of $H$, a finite
matrix. We similarly define the truncation $S_\ell$, a
finite matrix.

Let $X_\ell$ be the position operator on the finite box, amended to be
invertible.  More precisely, in $d=1$,
\begin{equation}
  \label{eq:position1d}
  X_\ell : \CC^\Bell \otimes \CC^{N/2}
  \to \CC^\Bell \otimes \CC^{N/2}
  \qquad
  X_\ell (\delta_x \otimes v) =
  \begin{cases}
    x (\delta_x \otimes v), & \ x\neq 0, \\
    \delta_0 \otimes v, & x=0.
  \end{cases}
\end{equation}
In $d=2$ we use the 
\emph{complex-valued position operator},
\begin{equation}
  \label{eq:position2d}
  X_\ell : \calH_\ell \to \calH_\ell
  \qquad
  X_\ell (\delta_{x,y} \otimes v) =
  \begin{cases}
    (x+\ii y) (\delta_{x,y} \otimes v), & \ x+\ii y \neq 0, \\
    \delta_{0,0} \otimes v, & x+\ii y=0.
  \end{cases}
\end{equation}

We define the finite-volume localizer by
\eql{\label{eq:finite-volume localizer} L_{\kappa,\ell}
  := \begin{cases} \br{1-\kappa}\begin{bmatrix}
      0 & S_\ell^\ast \\
      S_\ell & 0
    \end{bmatrix} + \kappa \begin{bmatrix}X_\ell & 0 \\ 0 & -X_\ell\end{bmatrix} & d=1 \\
    \br{1-\kappa}\begin{bmatrix}
      H_\ell & 0 \\ 0 & -H_\ell
    \end{bmatrix}+\kappa\begin{bmatrix}
      0 & X_\ell^\ast \\ X_\ell & 0
    \end{bmatrix} & d=2
  \end{cases}
} on $\calH_\ell$ or $\calH_\ell\oplus\calH_\ell$ respectively.  Note
that in the $d=2$ case there is a doubling of the Hilbert space on top of
the original degrees of freedom.

\begin{thm}[Loring--Schulz-Baldes]\label{thm:main}
  Let $H$ be a Hamiltonian obeying the locality condition
  \Cref{eq:locality} with locality constants $C<\infty$, $\mu>0$.  Assume it is
  gapped:
  \begin{equation}
    \label{eq:gap_def}
    \gap(H) := \norm{H^{-1}}^{-1}>0,
  \end{equation}
  and, in $d=1$, assume further it is
  chiral.  Let
  \begin{equation}
    \label{eq:constraint on maximal kappa}
    \kappa_\star:=\frac12\frac{\gap(H)^2}{\gap(H)^2+D^d},
    \qquad
    D:=\frac{C}{\cosh(\mu)-1}.      
  \end{equation}
  Then, for all
  \begin{equation}
    \label{eq:constraint_on_ell}
    \ell > \frac{\norm{H}}{\kappa_\star}
    \max\Set{2,1+\frac{8\norm{H}}{\gap(H)}},
  \end{equation}
  we have
  \eql{\frac12 \Sgntr(L_{\kappa_\star,\ell})=\begin{cases}
      \Zak(H) & d=1,\\
      \Chern(H) & d=2.
    \end{cases}}  
\end{thm}

We recall that for any invertible Hermitian matrix $M=M^\ast$, $\Sgntr(M)$ is defined as the number of positive minus the number of negative eigenvalues.

\subsection{Discussion and possible extensions}
The main contribution of this work is to review the proof of the Loring--Schulz-Baldes correspondence between bulk topological indices and the signature of the spectral localizer, in the chiral one-dimensional and integer quantum Hall settings, using only elementary functional analysis and basic properties of the spectral flow. In particular, the proof does not rely on $K$-theory or on abstract tools employed in \cite{LSB1,LSB2,LSB3}. Instead, the argument is based on three ingredients: (i) locality estimates for commutators with (functions of) the position operator, (ii) quantitative lower bounds on spectral gaps of block operators, and (iii) the stability and homotopy invariance of the spectral flow. This makes the spectral localizer construction more accessible to readers with a standard analysis background, while at the same time yielding explicit bounds on the admissible parameters $\kappa_\star$ and $\ell$ in \Cref{thm:main} in terms of the locality constants and the bulk gap.

A second conceptual point is the use of a bounded ``flattened'' position operator $f_\ell(X)$ in place of the unbounded position operator employed in the original localizer of \cite{Lor15,LSB1}. On the infinite-volume level, the resulting localizer $L_{\kappa,\ell}^\infty$ is a bounded perturbation of the chiral block operator or of the doubled Hamiltonian, which is technically convenient when invoking spectral-flow arguments and when passing to finite-volume truncations. The price to pay is a slightly more elaborate truncation scheme, but the boundedness of $f_\ell(X)$ allows us to keep uniform control of commutators and gaps as $\ell\to\infty$, which is crucial for obtaining quantitative estimates.

From the point of view of applications, our results justify the use of the spectral localizer as a robust finite-volume numerical diagnostic for one-dimensional chiral phases and two-dimensional Chern insulators. The fact that the signature of $L_{\kappa_\star,\ell}$ reproduces the bulk invariant for all sufficiently large boxes shows that no extrapolation to infinite volume is required, provided one works in the regime prescribed by \Cref{thm:main}. While our bounds on $\ell$ are probably far from optimal, they make the dependence on the gap and locality parameters completely explicit and may serve as a starting point for further optimization tailored to concrete models.

\medskip

Let us now comment on two natural extensions that are not treated in the present work.

\smallskip

\noindent\emph{The mobility gap regime.}  As recalled above, in the absence of a spectral gap but under a mobility gap (dynamical localization) assumption one can still define the bulk index in terms of the polar part $U=\polar(S)$ of the Fermi unitary, both in the chiral $d=1$ case and in the IQHE case \cite{EGS_2005,Shapiro20,BSS23}. In this regime the operator $\Lambda S \Lambda+\Lambda^\perp$ need not be Fredholm, so one is forced to work with $\Lambda U \Lambda+\Lambda^\perp$ instead. Our proof of \Cref{thm:main} already passes through the index formula for $U$ via \Cref{prop:main-index}, and the locality assumptions on $H$ imply suitable Schatten-class bounds on commutators $[U,f_\ell(X)]$ (compare \Cref{lem:commutator bounds with fl,lem:commutator bounds in two dimensions}). It is therefore plausible that the argument can be extended to the mobility-gap setting, replacing the spectral gap by dynamical localization estimates in the sense of \cite{EGS_2005,Shapiro20}. The main task would be to make sense of the notion of a spectral flow in the presence of dense pure point spectrum of Anderson localized states within the mobility gap. See also \cite{Stoiber2025} for a somewhat different approach.

\smallskip

\noindent\emph{Alternative boundary conditions and truncations.}  Throughout we work with Dirichlet-type truncations implemented by the partial isometries $J_\ell$, which amount to cutting the system to the box $\calB_\ell$ and extending wave-functions by zero outside. This choice is convenient because it decouples the interior and the exterior up to finite-rank errors, so that the infinite-volume localizer splits into a finite-volume part and a uniformly gapped complement. From a numerical perspective, however, periodic or twisted boundary conditions, as well as smoother spatial cut-offs, are often preferable. Our proof suggests that the specific choice of truncation is not essential: the key property is that the finite-volume localizer can be realized as a finite-rank perturbation of a uniformly gapped operator acting on the complement, with commutator bounds that are uniform in $\ell$. It should therefore be possible to adapt the argument to periodic boxes (with flux twists), to more general shapes, or to window functions with soft edges, at the cost of a slightly more involved bookkeeping of the off-diagonal terms. We do not attempt to develop such generalizations here.

\smallskip

\noindent\emph{Further directions.}  Beyond these two immediate extensions, several other directions appear promising. One is to treat higher spatial dimensions and symmetry classes beyond the complex chiral and unitary classes considered here, in particular real symmetry classes with $\ZZ_2$ indices, where spectral localizers for $\ZZ_2$ invariants have already been constructed in \cite{PSB_2016,LSB3}. Another is to adapt the present approach to Floquet systems, in which the fundamental object is a unitary time-evolution operator rather than a static Hamiltonian, and where bulk indices are again naturally expressed in terms of Fredholm indices of half-space compressions. Finally, it would be interesting to investigate whether the bulk-edge correspondence interpretation developed in \Cref{sec:bulk-edge correspondence perspective} can make sense in higher dimensions, in particular two dimensions.

\medskip

This paper is organized as follows. In \Cref{sec:proof} we prove
\Cref{thm:main} by first giving an outline of the proof in three key
lemmas and then proving those lemmas in subsections of
\Cref{sec:proof}.  The proof is performed in $d=1$ and $d=2$ cases on
an equal footing.  In \Cref{sec:bulk-edge correspondence perspective}
we revisit the construction in $d=1$ from the viewpoint of the
bulk-edge correspondence and interpret the spectral flow of the
localizer as an edge index of an associated domain-wall system.
Finally, several technical commutator estimates, gap bounds and a
short review of the spectral flow are collected in the appendices.

\section{The proof of \texorpdfstring{\Cref{thm:main}}{the main
    theorem}}
\label{sec:proof}

Let us consider $x\in\ZZ^2$ as a complex variable. Then we define the
map $f_\ell:\ZZ^d\to\CC$, for any $\ell\in\NN$, $d=1,2$:
\begin{equation}
  \label{eq:fell_def}
  \ZZ^d
  \ni x \mapsto \begin{cases}
    1 & x =0 \,, \\
    x & 1\leq |x|< \ell\,,\\
    \ell \frac{x}{|x|} & |x|\geq \ell\,;
  \end{cases}
\end{equation}
see \Cref{fig:enter-label}. We then consider an intermediate object, the infinite-volume localizer operator, as 
\begin{equation}
  \label{eq:infinite-volume localizer}
  L_{\kappa,\ell}^\infty := \begin{cases}
    \br{1-\kappa}\begin{bmatrix}
      0 & S^\ast \\
      S & 0
    \end{bmatrix} +
    \kappa \begin{bmatrix}f_\ell(X) & 0 \\ 0 & -f_\ell(X)\end{bmatrix} & d=1 \\
    \br{1-\kappa}\begin{bmatrix}
      H & 0 \\ 0 & -H
    \end{bmatrix}+\kappa\begin{bmatrix}
      0 & f_\ell(X)^\ast \\ f_\ell(X) & 0
    \end{bmatrix} & d=2
  \end{cases}  
\end{equation}
on $\calH$ or $\calH\oplus\calH$ respectively.  Here $X$ is the
lattice position operator on $\ell^2(\ZZ^d)$.

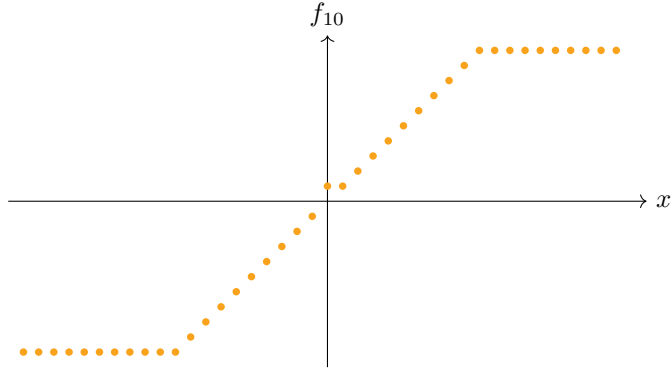
\begin{figure}
  \centering
  \begin{tikzpicture}[scale=0.2]
    \draw[->] (-21,0) -- (21,0) node[right] {\( x \)};
    \draw[->] (0,-11) -- (0,11) node[above] {\( f_{10} \)};
    
    \foreach \x/\y in {-20/-10,-19/-10,-18/-10,-17/-10,-16/-10,-15/-10,-14/-10,-13/-10,-12/-10,-11/-10, -10/-10, -9/-9, -8/-8, -7/-7, -6/-6, -5/-5, -4/-4, -3/-3, -2/-2, -1/-1, 0/1, 
      1/1, 2/2, 3/3, 4/4, 5/5, 6/6, 7/7, 8/8, 9/9, 10/10, 11/10, 12/10, 13/10, 14/10, 15/10, 16/10, 17/10, 18/10, 19/10} {
      \fill[YellowOrange] (\x,\y) circle (7pt);
    }
    
  \end{tikzpicture}
  \caption{A depiction of the function $f_{10}$ on $x\in[-20,20)\cap\ZZ$ in case $d=1$. Note $0\notin\im f_{10}$.}
  \label{fig:enter-label}
\end{figure}

We begin by connecting the topological invariant to the spectral flow in
\begin{lem}\label{lem:topological index is the spectral flow}
  For any $\ell\in\NN$ we have \eq{
    \operatorname{invariant}(H)
   &= \Sf([0,1]\ni\kappa\mapsto L_{\kappa,\ell}^\infty) } where \eq{
   \operatorname{invariant}(H) := \begin{cases}
      \Zak(H) & d=1\\
      \Chern(H) & d=2
    \end{cases}\,.
   }
\end{lem}
Outside the spectral localizer literature, this result can be compared with \cite{Phi_cmb96,Phillips,BoossBavnbekLeschPhillips2005UnboundedFredholmSpectralFlow}. Below we give a self-contained simple proof; in \cref{sec:bulk-edge correspondence perspective} we present a physical bulk-edge correspondence perspective on it for $d=1$.

Next we establish that the spectral flow in \Cref{lem:topological
  index is the spectral flow} may be started at $\kappa=\kappa_\star$
instead of $\kappa=0$, as $L_{\kappa,\ell}^\infty$ is invertible for
small $\kappa$.

\begin{lem}\label{lem:localizer is invertible for sufficiently small
    kappa}
  Recall $\kappa_\star$ and $D$ from \Cref{thm:main}.  Then for any
  $\kappa \in [0, 2\kappa_\star)$,
  \begin{equation}
    \inf_{\ell\in\NN}
    \label{eq:estimate_on_gap}
    \gap(L_{\kappa,\ell}^\infty) > 0,
  \end{equation}
  and therefore \eq{ \Sf([0,1]\ni\kappa\mapsto
    L_{\kappa,\ell}^\infty) = \Sf([\kappa_\star,1]\ni\kappa\mapsto
    L_{\kappa,\ell}^\infty)\,.  }
\end{lem}

The exclusion of $[0,\kappa_\star]$ from the spectral flow will allow
us to decouple the finite box from
the outer region in \Cref{lem:flow is insensitive to decoupling finite
  box with outside}.

Recall the subspace $\calH_\ell \subset \calH$ based on the finite box
$\Bell$; let $\calH_{\ell^c}$ be its orthogonal complement, which is
the infinite Hilbert space based on the outside of the box $\Bell$.  Finally,
let $J_{\ell^c} : \calH_{\ell^c} \to \calH$ be the inclusion operator.
We now define the (Dirichlet) restrictions of the operator
$L_{\kappa,\ell}^\infty$ to $\Bell$ and its complement by
\begin{equation}
  \label{eq:decoupled}
  L_{\kappa,\ell} := J_\ell^* L_{\kappa,\ell}^\infty  J_\ell, \qquad
  L_{\kappa,\ell^c} := J_{\ell^c}^* L_{\kappa,\ell}^\infty  J_{\ell^c}.
\end{equation}
We note that $L_{\kappa,\ell}$ in \eqref{eq:decoupled} is the same as
in \eqref{eq:finite-volume localizer} by the choice of $f_\ell$. Moreover,
\eq{
L_{\kappa,\ell}^\infty = L_{\kappa,\ell} \oplus L_{\kappa,\ell^c} + \text{connection pieces}\,.
}

Hence, even though, it is true that the spectral flow of a direct sum is the sum of the spectral flows, one must justify why the connection pieces may be ignored. We do this in 
\begin{lem}\label{lem:flow is insensitive to decoupling finite box
    with outside}
  For $\kappa_\star$ as in
  \eqref{eq:constraint on maximal kappa} and $\ell$ obeying
  \eqref{eq:constraint_on_ell}, the connection pieces do not contribute to the spectral flow and $L_{\kappa,\ell^c}$ is always invertible. Hence,
  \begin{align}
    \label{eq:flow_decouples}
    \Sf([\kappa_\star,1]\ni\kappa\mapsto L_{\kappa,\ell}^\infty)
    &= \Sf([\kappa_\star,1]\ni\kappa\mapsto L_{\kappa,\ell})
      + \Sf([\kappa_\star,1]\ni\kappa\mapsto L_{\kappa,\ell^c}) \\
    &= \Sf([\kappa_\star,1]\ni\kappa\mapsto L_{\kappa,\ell}).
      \label{eq:only_finite_flow}
  \end{align}
\end{lem}

Combining the previous three Lemmas, we immediately get that under the
assumptions of \Cref{thm:main},
\eql{ \Sf([\kappa_\star,1]\ni\kappa\mapsto
  L_{\kappa,\ell}) = \begin{cases}
    \Zak(H), & d=1,\\
    \Chern(H), & d=2.
  \end{cases}}
But the spectral flow of a finite Hermitian matrix
$L_{\kappa,\ell}$ equals half the difference of the signatures
between starting and ending points. At the ending point,
$\kappa=1$, the matrix has zero signature by construction since its eigenvalues
are exactly \eq{\Set{\pm \abs{f_\ell(x)} | 1\leq
    \abs{x}\leq\ell}\,.} The proof of \Cref{thm:main} is now
complete.

The rest of this section is devoted to the proof of
\Cref{lem:topological index is the spectral flow,lem:localizer is
  invertible for sufficiently small kappa,lem:flow is insensitive to
  decoupling finite box with outside}.

\subsection{Fredholm index via spectral
  flow: Proof of \Cref{lem:topological index is the spectral flow}}

We note that both $H$ and $H\oplus(-H)$ are invertible by hypothesis,
and by construction, so is $f_\ell(X)$. Moreover, $[H,f_\ell(X)]$ and
$[S,f_\ell(X)]$ are compact. Indeed, this is a consequence of the
exponential locality assumed on $H$, \Cref{eq:locality}, and
\Cref{lem:commutator bounds with fl} in one-dimension or
\Cref{lem:commutator bounds in two dimensions} in two-dimensions.

\Cref{lem:topological index is the spectral flow} now follows once we
establish the following abstract claim.

\begin{proposition}
  \label{prop:main-index}
  Let
  \begin{equation}
    \label{eq:H0-def}
    H_0 := \begin{bmatrix} 0 & B^* \\ B & 0 \end{bmatrix},
    \qquad
    H_1 := \begin{bmatrix} A & 0 \\ 0 & -A \end{bmatrix},
  \end{equation}
  with $B$ bounded invertible, $A$ bounded invertible self-adjoint, and
  $[A,B]$ compact.  Then
  \begin{equation}
    \label{eq:Sf-main}
    \Sf\big([0,1]\ni (1-t)H_0 + tH_1\big) \;=\; \ind\big(\Lambda B \Lambda + \Lambda^\perp\big),
  \end{equation}
  where $\Lambda:=\chi_{[0,\infty)}(A)$ is the spectral projection of $A$ corresponding to $[0,\infty)$.
\end{proposition}

The rest of this subsection is devoted to the proof of
\Cref{prop:main-index}. We do so in steps.

\begin{lem}[Spectral flow may be flattened]
  \label{lem:flatten_the_flow}
  Denoting by $T^\flat:=\polar(T)$ the polar part of an operator, i.e., $T^\flat\sqrt{T^\ast T}\equiv T$, we define the self-adjoint unitary operators
  \begin{equation}
    \label{eq:ourQ0Q1}
    Q_0 := \begin{bmatrix} 0 & B^{\flat*} \\ B^\flat & 0 \end{bmatrix},
    \qquad
    Q_1 := \begin{bmatrix} A^\flat & 0 \\ 0 & -A^\flat \end{bmatrix}.
  \end{equation}
  Then
  \begin{equation}
    \label{eq:step1}
    \Sf\big([0,1]\ni (1-t)H_0 + tH_1\big) \;=\;
    \Sf\big([0,1]\ni (1-t)Q_0 + tQ_1\big).
  \end{equation}
\end{lem}

\begin{proof}
  To establish \eqref{eq:step1}, we define the family
  \begin{equation}
    \label{eq:Hst-def}
    H_{t,s}
    :=
    (1-s)\big((1-t)H_0 + tH_1\big)
    +
    s\big((1-t)Q_0+ tQ_1\big),
    \qquad s,t\in[0,1].
  \end{equation}
  We observe that (with some details postponed)
  \begin{itemize}
  \item $H_{t,s}$ is Fredholm in all of $[0,1]^2$, therefore we can
    use homotopy property to replace the flow along the bottom edge
    (from $(t,s)=(0,0)$ to $(1,0)$), by the three-part path
    $(0,0)\to(0,1)\to(1,1)\to(1,0)$, see the right side of \Cref{fig:specQt}.
  \item $H_{1,s}$ is invertible for all $s\in[0,1]$ and therefore the
    spectral flow along the $(1,1)\to(1,0)$ part is zero.
  \item $H_{0,s}$ is invertible for all $s\in[0,1]$ and 
    therefore the spectral flow along the $(0,0)\to(0,1)$ part is zero.  
  \item Therefore the flow $(0,0)\to(1,0)$ is equal to the flow
    $(0,1)\to(1,1)$ which is precisely \eqref{eq:step1}.
  \end{itemize}
  
  We now provide the missing details.  The invertibility of $H_{1,s}$ follows from
  \begin{equation}
    \label{eq:H1s}
    H_{1,s} = A_s \oplus (-A_s),
    \qquad
    A_s = (1-s)A + sA^\flat 
    = A^\flat\big((1-s)|A| + s I\big),
  \end{equation}
  and $A_s$ being  obviously invertible for all $s\in[0,1]$.
  The invertibility of $H_{0,s}$ follows from
  \begin{equation}
    \label{eq:H0s}
    H_{0,s}^2 = B_s^*B_s \oplus B_sB_s^*,
    \qquad
    B_s = (1-s)B + s B^\flat 
    = B^\flat\big((1-s)|B| + s I\big),
  \end{equation}
  where $B_s$ is invertible.
  The Fredholm property of $H_{t,s}$ follows from the direct expansion
  \begin{equation}
    \label{eq:Hst-square}
    H_{t,s}^2
    = \big((1-t) H_{0,s} + t H_{1,s}\big)^2
    = (1-t)^2 H_{0,s}^2 + t^2 H_{1,s}^2
    + t(1-t)\big(H_{0,s}H_{1,s}+H_{1,s}H_{0,s}\big).
  \end{equation}
  The first two terms are invertible and therefore Fredholm.  The last
  term, using the notation $A_s$ and $B_s$ from \eqref{eq:H1s} and
  \eqref{eq:H0s}, is
  \begin{equation}
    \label{eq:Hanti}
    H_{0,s}H_{1,s}+H_{1,s}H_{0,s}
    = \begin{bmatrix}
      0 & [B_s,A_s]^* \\ [B_s,A_s] & 0
    \end{bmatrix},
  \end{equation}
  and therefore compact by our assumption that $[B,A]$ is compact and
  therefore  $[B,A^\flat]$, $[B^\flat,A]$ and $[B^\flat,A^\flat]$ are
  also compact.  Indeed, since $A,B$ are invertible, we may rewrite their polar part via the continuous functional calculus. However, the C-star algebra of operators with a compact commutator with a fixed operator is closed under the continuous functional calculus.
\end{proof}

We now want to evaluate the spectral flow in the right-hand side of
\eqref{eq:step1}.  We remark that $A^\flat \equiv 2\Lambda-\Id$.  Since $A$
is self-adjoint and the polar part of an operator is unitary by
definition, both $Q_0$ and $Q_1$ are self-adjoint unitaries.
Self-adjoint unitarity of $Q_0$ implies that $Q_0^2=\Id$ and that there
are only two eigenspaces, which correspond to the eigenvalues $\pm1$
and which we denote $E_{\pm1}(Q_0)$.

\begin{lem}[An index theorem for the flattened spectral flow]
  \label{lem:spectral-flow-Q}
  Let $Q_0$ and $Q_1$ be two self-adjoint unitary operators such that
  \begin{equation}
    \label{eq:fredholm-condition}
    0 \notin \sigma_{\mathrm{ess}}(Q_0 + Q_1).
  \end{equation}
  Then
  \begin{equation}
    \label{eq:Sf-equals-index-P}
    \Sf\big([0,1] \ni t \mapsto (1-t)\,Q_0 + t\,Q_1 \big)
    \;=\;
    \dim\big(E_{-1}(Q_0)\cap E_{1}(Q_1)\big) 
    \;-\;
    \dim\big(E_{1}(Q_0)\cap E_{-1}(Q_1)\big).
  \end{equation}
\end{lem}

\begin{rem}
  The pair of self-adjoint unitaries $(Q_0,Q_1)$
  satisfying~\eqref{eq:fredholm-condition} is called a Fredholm
  pair and the integer in the right-hand side of
  \eqref{eq:Sf-equals-index-P} is called the index of the pair
  $(Q_0,Q_1)$.

  Note moreover that the particular $2\times2$ block operators $Q_0$ and $Q_1$  defined in \eqref{eq:ourQ0Q1} indeed form a
  Fredholm pair: calculating $(Q_0+Q_1)^2$ as a $2\times 2$ block operator, we find the result is a sum of two invertible positive operators on each diagonal block, and commutators on the off-diagonal. However, the commutators are, by hypothesis, compact.
\end{rem}

\begin{proof}[Proof of \cref{lem:spectral-flow-Q}]
  On the interval $[0,\frac12)$, we view $Q_t := Q_0 + t(Q_1-Q_0)$ as
  a perturbation of $Q_0$ of norm less than 1, therefore it is
  Fredholm and has empty kernel.  On the interval $(\frac12, 1]$, we
  similarly view $Q_t = Q_1 + (1-t)(Q_0 - Q_1)$ as a perturbation of
  $Q_1$ of norm less than 1.  At $t=\frac12$, $Q_t$ is Fredholm by
  \eqref{eq:fredholm-condition}.  We conclude that $t=\frac12$ is
  the only location where the eigenvalue curves can cross 0, see
  \Cref{fig:specQt}.

  \usetikzlibrary{decorations.markings}
  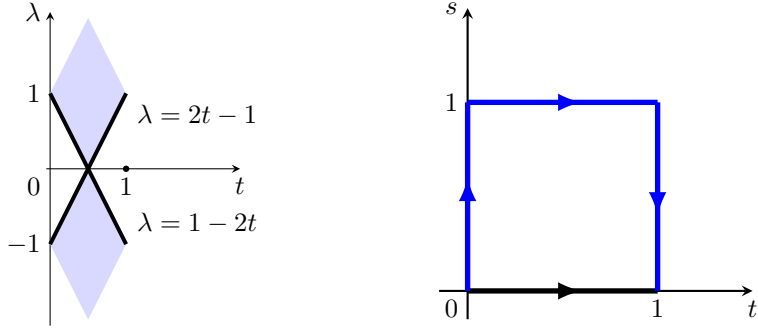
\begin{figure}
    \centering
    \begin{tikzpicture}[scale=1, >=stealth]
      \draw[->] (-0.05,0) -- (2.5,0) node[below] {$t$};
      \draw[->] (0,-2.08) -- (0,2.08) node[left] {$\lambda$};

      \node[below left] at (0,0) {$0$};
      \fill (1,0) circle (1.2pt) node[below] {$1$};
      \node[left]  at (0,1) {$1$};
      \node[left]  at (0,-1) {$-1$};

      \fill[blue!15] (0.5,2) -- (1,1) -- (0.5,0) -- (0,1) -- cycle;
      \fill[blue!15] (0.5,-2) -- (1,-1) -- (0.5,0) -- (0,-1) -- cycle;
      
      \draw[line width=1.5pt] (0,-1) -- (1,1) node[below right] {$\lambda = 2t - 1$};
      \draw[line width=1.5pt] (0, 1) -- (1,-1) node[above right] {$\lambda = 1 - 2t$};
    \end{tikzpicture}
    \hspace{2cm}
    \begin{tikzpicture}[scale=2.5, >=stealth,decoration={markings,
        mark=at position 0.6 with {\arrow{latex};}}]
      \draw[->, thick] (-0.15,0) -- (1.5,0) node[below] {$t$};
      \draw[->, thick] (0,-0.15) -- (0,1.5) node[left] {$s$};
      \node[below left] at (0,0) {$0$};
      \node[below] at (1,0) {$1$};
      \node[left]  at (0,1) {$1$};

      \draw[line width=2pt,black,postaction=decorate] (0,0) -- (1,0);
      \draw[line width=2pt,blue,postaction=decorate] (0,0) -- (0,1);
      \draw[line width=2pt,blue,postaction=decorate] (0,1) -- (1,1);
      \draw[line width=2pt,blue,postaction=decorate] (1,1) -- (1,0);
    \end{tikzpicture}

    \caption{(Left) The shaded area indicates where the spectrum of $Q_t$
      must lie for $t\in[0,1]$.  The only crossing of $0$ can occur at
      $t=\tfrac12$.  Black lines indicate the two eigenspaces we know
      explicitly, namely $E_{-1}(Q_0)\cap E_{1}(Q_1)$ and $E_{1}(Q_0)\cap
      E_{-1}(Q_1)$. (Right) The two paths for computing the spectral flow of $H_{t,s}$ defined by \eqref{eq:Hst-def}.}
    \label{fig:specQt}
  \end{figure}

  It can be checked explicitly that, for any $t$, $E_{-1}(Q_0)\cap E_{1}(Q_1)$ belongs to the eigenspace of $Q_t$ with the eigenvalue $2t-1$.  Similarly, $E_{1}(Q_0)\cap E_{-1}(Q_1)$ is a subspace of the eigenspace of $Q_t$ with the eigenvalue $1-2t$.  Once we show that
  \begin{equation}
    \label{eq:kerQQ}
    \ker(Q_0+Q_1) = \big(E_{-1}(Q_0)\cap E_{1}(Q_1)\big) 
    \oplus \big(E_{1}(Q_0)\cap E_{-1}(Q_1)\big),
  \end{equation}
  we can conclude that $2t-1$ and $1-2t$ are the only curves crossing 0 and the spectral flow is equal to the number of curves going up, i.e.
  $\dim\big(E_{-1}(Q_0)\cap E_{1}(Q_1)\big)$ minus the number of curves going down, i.e. $\dim\big(E_{1}(Q_0)\cap E_{-1}(Q_1)\big)$, yielding \eqref{eq:Sf-equals-index-P}.

  To show \eqref{eq:kerQQ}, we check the inclusion $\supset$ by an explicit computation.  To establish $\subset$, let $x$ be such that $(Q_0+Q_1)x=0$.  We write
  \begin{equation}
    x = \frac12(1+Q_0)x + \frac12(1-Q_0)x =: x_+ + x_-.
  \end{equation}
  Using $Q_0^2 = Q_1^2 = 1$, we compute
  \begin{equation}
    (1-Q_0)x_+ = \frac12(1-Q_0)(1+Q_0)x = \frac12(1-Q_0^2)x = 0, 
  \end{equation}
  and
  \begin{equation}
    (1+Q_1)x_+ = \frac12(1+Q_1+Q_0 + Q_1Q_0)x 
    = \frac12(1+Q_1)(Q_1+Q_0)x = 0. 
  \end{equation}
  Therefore $x_+ \in E_{1}(Q_0)\cap E_{-1}(Q_1)$ and similarly $x_-\in E_{-1}(Q_0)\cap E_{1}(Q_1)$.
\end{proof}
We can finally prove the main index theorem.
\begin{proof}[Proof of \Cref{prop:main-index}]
  All that is left is to compute the right-hand side of
  \eqref{eq:Sf-equals-index-P} when $Q_0$ and $Q_1$ are given by
  \eqref{eq:ourQ0Q1}.  We write
  \begin{equation}
    \label{eq:P0P1-top}
    \Id+Q_0 \;=\; 
    \begin{bmatrix} 
      \Id & B^{\flat*} \\ B^\flat & \Id 
    \end{bmatrix},
    \qquad
    \Id-Q_1 \;=\;
    2\begin{bmatrix} \Lambda^\perp & 0 \\ 0 & \Lambda \end{bmatrix},
  \end{equation}
  where we recall that $A^\flat = 2\Lambda - 1$.

  We compute
  \begin{equation}
    \label{eq:ran-P0}
    E_{-1}(Q_0) = \ker(\Id+Q_0)
    = \big\{(x,-B^\flat x): x\in\calH\big\}
    = \big\{\big(B^{\flat*} y, -y\big): y\in\calH\big\},
  \end{equation}
  \begin{equation}
    \label{eq:ker-P1}
    E_1(Q_1) = \ker(\Id-Q_1)
    \;=\; \mathrm{Ran}\,\Lambda \oplus \ker\Lambda,
  \end{equation}
  and therefore
  \begin{equation}
    \label{eq:ranP0-cap-kerP1}
    E_{-1}(Q_0) \cap E_1(Q_1)
    \;=\;
    \big\{(x,-B^\flat x): x\in \mathrm{Ran}\,\Lambda,\;
    B^\flat x\in \ker\Lambda\big\}
    \simeq
    \ker\big(\Lambda B^\flat\restriction_{\mathrm{Ran}\,\Lambda}\big).
  \end{equation}
  Similarly,
  \begin{equation}
    \label{eq:ker-P0}
    \ker(\Id-Q_0) \;=\; \big\{\big(B^{\flat*} x, x\big) : x\in\calH\big\},
    \qquad
    \ker(\Id+Q_1) \;=\; \ker\Lambda \oplus \mathrm{Ran}\,\Lambda,
  \end{equation}
  and therefore
  \begin{equation}
    \label{eq:kerP0-cap-ranP1}
    E_{1}(Q_0) \cap E_{-1}(Q_1)
    \;=\;
    \big\{\big( B^{\flat*} x,x): x\in \mathrm{Ran}\,\Lambda,\; B^{\flat*}x\in \ker\Lambda\big\}
    \simeq
    \ker\big(\Lambda B^{\flat*}\restriction_{\mathrm{Ran}\,\Lambda}\big).
  \end{equation}
  
  By \Cref{lem:spectral-flow-Q},
  \begin{equation}
    \label{eq:sf-top-equals-index}
    \Sf(Q_t)
    \;=\;
    \dim\ker\big(\Lambda B^\flat\restriction_{\mathrm{Ran}\,\Lambda}\big)
    \;-\;
    \dim\ker\big(\Lambda B^{\flat*}\restriction_{\mathrm{Ran}\,\Lambda}\big)
    \;=\;
    \ind\big(\Lambda B^\flat \Lambda + \Lambda^\perp\big),
  \end{equation}
  where the second equality follows from the decomposition
  $\calH = \mathrm{Ran}\,\Lambda\oplus\ker\Lambda$ being reducing for
  the operator
  $\Lambda B^\flat \Lambda + \Lambda^\perp = \Lambda B^\flat \oplus
  I$.

  Finally,
  \begin{equation}
    \label{eq:final-index-S}
    \ind\big(\Lambda B^\flat \Lambda + \Lambda^\perp\big)
    \;=\;
    \ind\big(\Lambda B \Lambda + \Lambda^\perp\big).
  \end{equation}
  by homotopy invariance of the Fredholm index. Indeed, $\Lambda B_t \Lambda+\Lambda^\perp$ is Fredholm as soon as $[\Lambda,B_t]$ is compact and $B_t$ is invertible. This, however, is guaranteed for the same reasons explained above for $B_t:=(1-t)B+tB^\flat$ for $t\in[0,1]$.
\end{proof}

\subsection{Proof of \Cref{lem:localizer is invertible for
    sufficiently small kappa}}

We will establish the estimate
\begin{equation}
  \label{eq:gap_est}
  \inf_{\ell\in\NN} \gap\br{L_{\kappa,\ell}^\infty} \geq
  \sqrt{\br{1-\kappa}^2\gap(H)^2-\kappa\br{1-\kappa}D^d},
  \qquad d = 1,2,
\end{equation}
from which the claim of the Lemma follows.

In $d=1$, we use estimate in \Cref{lem:block_traceless} with 
$A=\kappa f_\ell(X)$
and $B = (1-\kappa)S$.  The commutator $[A,B]$ is
uniformly bounded by \Cref{lem:commutator bounds with
  fl}, yielding $\norm{[A,B]}\leq
\kappa(1-\kappa)\frac{C}{\cosh\br{\mu}-1}=:\kappa(1-\kappa)D$.  We use
the fact that $\gap(S)=\gap(H)$ to arrive to \eqref{eq:gap_est}.

If $d=2$, we apply \Cref{lem:block_traceless} with $A=(1-\kappa)H$ and
$B=\kappa f_\ell(X)$. Again, the commutator is uniformly bounded by
\Cref{lem:commutator bounds in two dimensions}, so,
$\norm{[A,B]}\leq \kappa(1-\kappa)\br{\frac{C}{\cosh\br{\mu}-1}}^2 =:
\kappa(1-\kappa)D^2$ and estimate \eqref{eq:gap_est} results.

Finally, the quantity in the square root is strictly positive for all values
of $\kappa$ strictly smaller than
\begin{equation*}
  \frac{\gap(H)^2}{\gap(H)^2+D^d} \ =:\  2\kappa_\star,  
\end{equation*}
completing the proof.

\subsection{Proof of \Cref{lem:flow is insensitive to decoupling
    finite box with outside}}

With respect to the orthogonal decomposition
\begin{equation*}
  \calH = \calH_\ell \oplus \calH_{\ell^c},
\end{equation*}
we define
\begin{equation}
  \label{eq:Moperator}
  M_{\kappa,s} =
  \begin{bmatrix}
    J_\ell^* L_{\kappa,\ell}^\infty J_\ell
    & s \left(J_\ell^* L_{\kappa,\ell}^\infty J_{\ell^c}\right) \\
    s\left(J_{\ell^c}^* L_{\kappa,\ell}^\infty J_\ell\right)
    & J_{\ell^c}^* L_{\kappa,\ell}^\infty J_{\ell^c},  
  \end{bmatrix}
\end{equation}
so that $M_{\kappa,s=1} = L_{\kappa,\ell}^\infty$ and $M_{\kappa,s=0}
= L_{\kappa,\ell} \oplus L_{\kappa,\ell^c}$.

In the region $[\kappa_\star,1] \times [0,1]$, the operator
$M_{\kappa,s}$ is \emph{essentially} gapped: the only block of
non-finite rank is $J_{\ell^c}^* L_{\kappa,\ell}^\infty J_{\ell^c}$
which is invertible by \eqref{eq:gap of off-diagonal piece} below.  By
homotopy invariance of the spectral flow, see
\cref{thm:spec_flow_basics}, the claim in equation
\eqref{eq:flow_decouples} of the Lemma will be established if we can
show that the flow of $M_{\kappa,s}$ in the variable $s$ is zero at
both $\kappa=\kappa_\star$ and $\kappa=1$.

At $\kappa=1$, the operator $L_{\kappa,\ell}^\infty$  contains only position terms
$f_\ell(X)$ (see
\eqref{eq:infinite-volume localizer}), which commute (in the appropriate sense) with $J_\ell$
and therefore $J_{\ell^c}^* L_{\kappa,\ell}^\infty J_\ell = 0$ because
$J_{\ell^c}^* J_\ell = 0$.  In other words, $M_{1,s}$ is independent
of $s$ and generates no flow.

We now turn to the $[0,1]\ni s$ flow at $\kappa=\kappa_\star$, for which we intend
to use \Cref{lem:invertible_cut}.  Note that the
off-diagonal pieces are bounded: since $f_\ell(X)$ commutes with $J_\ell$ 
\eq{
  \norm{J_{\ell^c}^* L_{\kappa_\star,\ell}^\infty J_\ell} \leq (1-\kappa_\star) \norm{H} \leq \norm{H}.
}
For the $(2,2)$-block we have
\begin{equation}
  \label{eq:gap of off-diagonal piece}
  \gap(J_{\ell^c}^*\,  L_{\kappa_\star,\ell}^\infty\, J_{\ell^c}) \geq
  \kappa_\star \ell - \norm{H}.
\end{equation}
Indeed, outside the box $\Bell$, the multiplication operator $f_\ell(X)$ in
\eqref{eq:infinite-volume localizer} is a multiplication by a number of
absolute value $\ell$, with the spectral gap equal to $\ell$.  We
therefore use the estimate $\gap(A+B) \geq \gap(A) - \|B\|$, see
\Cref{lem:basic_gap}, to obtain \eqref{eq:gap of off-diagonal piece}.

For the gap of $M_{\kappa_\star, s}$ at $s=1$, we use the estimate
\eqref{eq:gap_est}, which, with $\kappa=\kappa_\star < \frac12$ yields
\begin{equation*}
  \gap\left(M_{\kappa_\star, s}\right)
  = \gap\left(L_{\kappa_\star,\ell}^\infty \right)
  \geq \frac12 \gap(H).
\end{equation*}

The condition in \cref{lem:invertible_cut} becomes
\begin{equation}
  \label{eq:equiv_ell_cond}
  \kappa_\star \ell - \norm{H}
  > \norm{H} \max\br{\Set{1,\frac{\norm{H}}{\frac18 \gap\br{H}}}},  
\end{equation}
which is equivalent to \eqref{eq:constraint_on_ell}.

We have thus established \eqref{eq:flow_decouples}.  To pass to
\eqref{eq:only_finite_flow} we observe that, analogously to
\eqref{eq:gap of off-diagonal piece}, we have for $\kappa\geq\kappa_\star$
\begin{equation*}
  \gap\left(L_{\kappa,\ell^c}\right)
  = \gap(J_{\ell^c}^*\,  L_{\kappa,\ell}^\infty\, J_{\ell^c}) \geq
  \kappa \ell - \norm{H} > \kappa_\star \ell - \norm{H},
\end{equation*}
and therefore the operator $L_{\kappa,\ell^c}$ is invertible by
\eqref{eq:equiv_ell_cond} and generates no flow in
\eqref{eq:flow_decouples}.  The proof of our lemma is complete.

\section{The bulk-edge correspondence perspective on \Cref{lem:topological index is the spectral flow}}\label{sec:bulk-edge correspondence perspective}
Here we want to shed some light on a physical interpretation of the parameter $\kappa$ and why the scheme of \Cref{lem:topological index is the spectral flow} works, beyond the rather abstract \Cref{prop:main-index}.  We focus strictly on the $1$-dimensional case; we are unaware of how to make sense of what follows for $d=2$. We shall recast the localizer as an effective two-dimensional model.

To that end, define the infinite-volume localizer now as
\eql{
  L_{\kappa} := \begin{bmatrix}
    0 & S^\ast \\ S & 0
  \end{bmatrix} + \kappa\sgn(X)\otimes\sigma_3\qquad(\kappa\in[0,\infty))
} where $S$ is the chiral block of $H$ as defined in \Cref{eq:chiral constraint}. Note this notation differs from $L_{\kappa,\ell}^\infty$ in the previous section since now the interpolation parameter goes from $0$ to $\infty$, which should not matter too much. Moreover, in this definition we have set $\ell=1$ compared with $L_{\kappa,\ell}^\infty$. The plan is to first establish equivalence with the spectral flow and only \emph{after} change from $\ell=1$ to larger values.

Within this section, our goal is merely to give another plausible explanation for \Cref{lem:topological index is the spectral flow} via the bulk-edge correspondence. I.e., we will show 
\begin{lem}\label{lem:BEC proof of index theorem}
  \eq{
    \findex\br{\bbLambda S }  = \Sf\br{[0,\infty)\ni\kappa\mapsto L_\kappa}
  } where $\bbLambda S \equiv \Lambda S \Lambda + \Lambda^\perp$.
\end{lem}
The strategy for the proof is as follows. We interpret $\mathbb{R}\ni\kappa\mapsto L_{\kappa}$
as the fiber of a two-dimensional \emph{domain-wall Hamiltonian}, i.e., its
second coordinate $\xi$ has conjugate momentum variable $\kappa$.
The domain-wall Hamiltonian interpolates between two copies: the one
on the left-hand side $L_{\kappa}^{-}$ and the one on the right-hand
side $L_{\kappa}^{+}$, see \Cref{fig:domain wall system}. The two operators are given by 
\eql{\label{eq:the two sides of the domain wall localizer}
  L_{\kappa}^{\pm}  :=  H\pm\kappa\sigma_{3}\,.
} 

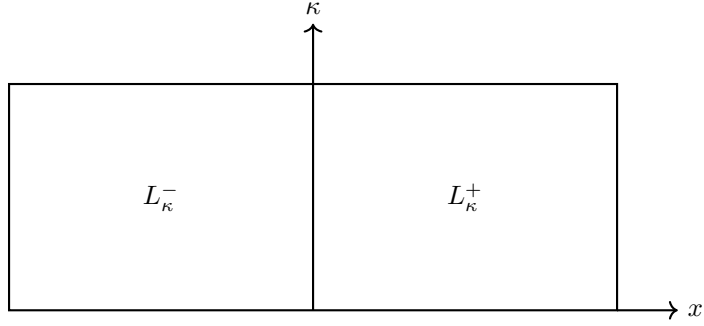
\begin{figure}
  \centering
  \begin{tikzpicture}[thick]

    \draw (-4,0) -- (-4,3) -- (4,3) -- (4,0) -- cycle;

    \draw (0,0) -- (0,3);

    \draw[->] (0,3) -- (0,3.8) node[above] {$\kappa$};

    \draw[->] (4,0) -- (4.8,0) node[right] {$x$};

    \node at (-2,1.5) {$L_\kappa^{-}$};
    \node at ( 2,1.5) {$L_\kappa^{+}$};

  \end{tikzpicture}
  \caption{The domain-wall system in mixed real-space for $x$ momentum space for $\kappa$.}
  \label{fig:domain wall system}
\end{figure}

The rest of this section is devoted to the proof of \Cref{lem:BEC proof of index theorem}.

As defined above, we work with \cref{eq:the two sides of the domain wall localizer} and interpret this operator as the fiber of an operator on $\left(\ell_{x}^{2}\left(\mathbb{Z}\right)\otimes L_{\xi}^{2}\left(\mathbb{R}\right)\right)\otimes\mathbb{C}^{2N}$
with $\kappa$ as the momentum variable conjugate to the second coordinate
$\xi$, namely, the fiber of the operator 
\[
  H\pm\sigma_{3}\ii\partial_{\xi}\,.
\]
This operator is \emph{always }gapped. Indeed,
\eq{
  \left(L_{\kappa}^{\pm}\right)^{2} = \begin{bmatrix}\kappa^{2}+\left|S\right|^{2} & 0\\
    0 & \kappa^{2}+\left|S^{\ast}\right|^{2}
  \end{bmatrix} \geq  \kappa^{2}+\norm{S^{-1}}^{-2}\Id\geq \norm{S^{-1}}^{-2}\Id\,.
}
Cf. with $L_{\kappa}$ which is always \emph{essentially }gapped but
whose gap should actually close at some $\kappa$ so as to yield a
non-zero spectral flow.

The gap allows us to define its flat version $\Sigma_{\kappa}^{\pm}$,
given by 
\begin{eqnarray*}
  \Sigma_{\kappa}^{\pm} = \sgn\left(L_{\kappa}^{\pm}\right)= \begin{bmatrix}\kappa^{2}+\left|S\right|^{2} & 0\\
    0 & \kappa^{2}+\left|S^{\ast}\right|^{2}
  \end{bmatrix}^{-\frac{1}{2}}\begin{bmatrix}\pm\kappa & S^{\ast}\\
    S & \mp\kappa
  \end{bmatrix}\,.
\end{eqnarray*}
Note that the $\kappa\to\pm\infty$ boundary conditions are
\begin{eqnarray*}
  \Sigma_{\kappa}^{\pm} & \stackrel{\kappa\to\pm\infty}{\approx} & \pm\sgn\left(\kappa\right)\sigma_{3}
\end{eqnarray*}
so that we \emph{do not }have a compactification of $\kappa\in\mathbb{R}$.
We emphasize that this is not an issue. The Chern number is usually
associated with a vector bundle (the Bloch bundle) over a two-dimensional
torus when a system has full translation invariance and its two space
coordinates are defined over $\mathbb{Z}^{2}$, whence the Pontryagin
dual becomes $\mathbb{T}^{2}$, the Brillouin torus. The Hall conductivity,
and the topological index associated to it, however, is still well-defined
even without translation invariance, and when the real-space coordinates
are defined over $\mathbb{R}^{2}$, and in fact, also when one of
them is defined over $\mathbb{Z}$ and the other over $\mathbb{R}$,
which is the particular case at hand. 

We also note in passing that this problem could have been avoided all-together in two different ways:
\begin{itemize}
\item Firstly, interpolate using trigonometric instead of linear functions so as to obtain a periodic model on $\mathbb{S}^1$ instead of $\RR$:
  \eql{
    L_{\kappa} := \cos(\kappa)\begin{bmatrix}
      0 & S^\ast \\ S & 0
    \end{bmatrix} + \sin(\kappa)\sgn(X)\otimes\sigma_3\qquad(\kappa\in[0,\pi/2))
  }
  
\item The model is still not periodic on $\kappa\in(-\pi/2,\pi/2)$ but may be extended in a periodic way which does not affect the overall flow (for example by connecting both ends to a mutual $\Id\otimes\sigma_1$).
\end{itemize}

This makes the conceptual connection to the bulk-edge correspondence on $\ell^2(\ZZ^2)$ easier, but the calculations to follow much harder. Since this section is anyway only a heuristic we proceed with the model on $\ell^2(\ZZ)\otimes L^2(\RR)$.

Hence, for us, the second
coordinate $\kappa$ is to be interpreted as a momentum coordinate.
As such the appropriate ``double-commutator'' Kubo formula for the
Hall conductivity in this case is given by 
\[
  \Chern\left(L_{\kappa}^{\pm}\right)=\ii\int_{\kappa\in\mathbb{R}}\dif{\kappa}\tr_{\ell^2(\ZZ)}P_{\kappa}^{\pm}\left[\left[\Lambda,P_{\kappa}^{\pm}\right],-\ii\partial_{\kappa}P_{\kappa}^{\pm}\right]
\]
where $P_{\kappa}^{\pm}\equiv\frac{1}{2}\left(\Id-\Sigma_{\kappa}^{\pm}\right)$
is the associated ``Fermi'' projection (the Fermi energy in this
case is chosen as zero).

\subsection{Flattening operators}

\begin{claim}
  We may replace $S$ with $U:=\operatorname{polar}\left(S\right)\equiv S\left|S\right|^{-1}$
  in the localizer's Chern number calculation, i.e., 
  \begin{eqnarray*}
    \Chern\left(\begin{bmatrix}0 & S^{\ast}\\
        S & 0
      \end{bmatrix}\pm\kappa\sigma_{3}\right) & = & \Chern\left(\begin{bmatrix}0 & U^{\ast}\\
        U & 0
      \end{bmatrix}\pm\kappa\sigma_{3}\right)\,.
  \end{eqnarray*}
\end{claim}
The statement is very similar to \Cref{lem:flatten_the_flow} and
  \eqref{eq:final-index-S}.

\begin{proof}
  As we have seen, 
  \begin{eqnarray*}
    \left(\begin{bmatrix}0 & S^{\ast}\\
        S & 0
      \end{bmatrix}\pm\kappa\sigma_{3}\right)^{2} & = & \begin{bmatrix}\kappa^{2}+\left|S\right|^{2} & 0\\
      0 & \kappa^{2}+\left|S^{\ast}\right|^{2}
    \end{bmatrix}\geq\norm{S^{-1}}^{-2}\Id
  \end{eqnarray*}
  which means that as we deform 
  \begin{eqnarray*}
    \left[0,1\right]\ni t & \mapsto & S_{t}:=\left(1-t\right)U\left|S\right|+tU=U\left[\left(1-t\right)\left|S\right|+t\Id\right]
  \end{eqnarray*}
  we keep the gap of 
  \[
    \begin{bmatrix}0 & S_{t}^{\ast}\\
      S_{t} & 0
    \end{bmatrix}\pm\ii\partial_{\kappa}\sigma_{3}
  \]
  remains open, since 
  \[
    \left(1-t\right)\left|S\right|+t\Id\geq\min(1,\norm{S^{-1}}^{-1})\Id\,.
  \]
  Locality is furthermore clearly conserved since $S_{t}$ is local,
  since $U$ is, thanks to the gap assumption. Since the gap and locality
  remain open, the path 
  \begin{eqnarray*}
    \left[0,1\right]\ni t & \mapsto & \begin{bmatrix}0 & S_{t}^{\ast}\\
      S_{t} & 0
    \end{bmatrix}\pm\ii\partial_{\xi}\sigma_{3}
  \end{eqnarray*}
  is a norm continuous uniform (in $t$) path of local gapped operators
  on $\ell^{2}\left(\mathbb{Z}\right)\otimes L^{2}\left(\mathbb{R}\right)$, and as such, since the Chern number is
  given by an index of a Fredholm operator built out of $\sgn\left(\begin{bmatrix}0 & S_{t}^{\ast}\\
      S_{t} & 0
    \end{bmatrix}\pm\kappa\sigma_{3}\right)$, namely, 
  \[
    \Chern\left(\begin{bmatrix}0 & S_{t}^{\ast}\\
        S_{t} & 0
      \end{bmatrix}\pm\kappa\sigma_{3}\right)=\findex_{\ell^{2}\left(\mathbb{Z}\right)\otimes L^{2}\left(\mathbb{R}\right)}\left(P_{t}LP_{t}+P_{t}^{\perp}\right)
  \]
  where the $t$-dependent Fermi projection is
  \[
    P_{t}:=\chi_{\left(-\infty,0\right)}\left(\begin{bmatrix}0 & S_{t}^{\ast}\\
        S_{t} & 0
      \end{bmatrix}\pm\ii\partial_{\xi}\sigma_{3}\right)
  \]
  and the mixed Laughlin flux insertion operator is
  \begin{eqnarray*}
    L & := & \operatorname{polar}\left(X+\ii\Xi\right)\,.
  \end{eqnarray*}
  The map $t\mapsto P_{t}$ preserves all these properties, so by the
  stability of the Fredholm index w.r.t. norm continuous deformations,
  the Chern number remains constant throughout this path.
\end{proof}
\begin{cor}
  As a result, to show 
  \[
    \findex\left(\mathbb{\Lambda}S\right)=\pm\Chern\left(\begin{bmatrix}0 & S^{\ast}\\
        S & 0
      \end{bmatrix}\pm\kappa\sigma_{3}\right)\,,
  \]
  since $\findex\left(\mathbb{\Lambda}S\right)=\findex\left(\mathbb{\Lambda}U\right)$
  \cite{Graf_Shapiro_2018_1D_Chiral_BEC}, we can just as well show that 
  \[
    \findex\left(\mathbb{\Lambda}U\right)=\pm\Chern\left(\begin{bmatrix}0 & U^{\ast}\\
        U & 0
      \end{bmatrix}\pm\kappa\sigma_{3}\right)\,.
  \]
\end{cor}

\subsubsection{Calculation of the Chern number}

We now show 
\[
  \findex\left(\mathbb{\Lambda}U\right)=\pm\Chern\left(\begin{bmatrix}0 & U^{\ast}\\
      U & 0
    \end{bmatrix}\pm\kappa\sigma_{3}\right)
\]
using the double-commutator formula for the Chern number:
\[
  \Chern\left(\begin{bmatrix}0 & U^{\ast}\\
      U & 0
    \end{bmatrix}\pm\kappa\sigma_{3}\right)=-\frac{1}{8}\int_{\kappa\in\mathbb{R}}\dif{\kappa}\tr_{\ell^{2}\left(\mathbb{Z}\right)}\left(\Sigma_{\kappa}^{\pm}\left[\left[\Lambda,\Sigma_{\kappa}^{\pm}\right],\partial_{\kappa}\Sigma_{\kappa}^{\pm}\right]\right)\,.
\]

We have 
\begin{eqnarray*}
  \left[\Lambda,\Sigma_{\kappa}^{\pm}\right]  \equiv  \left[\begin{bmatrix}\Lambda & 0\\
      0 & \Lambda
    \end{bmatrix},\Sigma_{\kappa}^{\pm}\right]= \begin{bmatrix}0 & \left(\kappa^{2}+1\right)^{-\frac{1}{2}}\left[\Lambda,U^{\ast}\right]\\
    \left(\kappa^{2}+1\right)^{-\frac{1}{2}}\left[\Lambda,U\right] & 0
  \end{bmatrix}
\end{eqnarray*} and
\eq{
  \partial_{\kappa}\Sigma_{\kappa}^{\pm} 
  = \begin{bmatrix}
    \pm(\kappa^{2}+1)^{-3/2} & -\kappa(\kappa^{2}+1)^{-3/2}U^{\ast}\\
    -\kappa(\kappa^{2}+1)^{-3/2}U & \mp(\kappa^{2}+1)^{-3/2}
  \end{bmatrix}}
and so 
\begin{eqnarray*}
  \left[\left[\Lambda,\Sigma_{\kappa}^{\pm}\right],\partial_{\kappa}\Sigma_{\kappa}^{\pm}\right]
  =
  \begin{bmatrix}-\left(\kappa^{2}+1\right)^{-\frac{3}{2}}\kappa\left(\left[\Lambda,U^{\ast}\right]U-U^{\ast}\left[\Lambda,U\right]\right) & \mp2\left(\kappa^{2}+1\right)^{-\frac{3}{2}}\left[\Lambda,U^{\ast}\right]\\
    \pm2\left(\kappa^{2}+1\right)^{-\frac{3}{2}}\left[\Lambda,U\right] & -\left(\kappa^{2}+1\right)^{-\frac{3}{2}}\kappa\left(\left[\Lambda,U\right]U^{\ast}-U\left[\Lambda,U^{\ast}\right]\right)
  \end{bmatrix}\,.
\end{eqnarray*}
But $\left[\Lambda,U^{\ast}\right]U-U^{\ast}\left[\Lambda,U\right] =2U^{\ast}\left[U,\Lambda\right] $ and similarly with $U^\ast$,
and hence
\begin{eqnarray*}
  \left[\left[\Lambda,\Sigma_{\kappa}^{\pm}\right],\partial_{\kappa}\Sigma_{\kappa}^{\pm}\right] & = & 2\left(\kappa^{2}+1\right)^{-\frac{3}{2}}\begin{bmatrix}\kappa U^{\ast}\left[\Lambda,U\right] & \mp\left[\Lambda,U^{\ast}\right]\\
    \pm\left[\Lambda,U\right] & \kappa U\left[\Lambda,U^{\ast}\right]
  \end{bmatrix}\,.
\end{eqnarray*}
From this we find 
\begin{eqnarray*}
  \Sigma_{\kappa}^{\pm}\left[\left[\Lambda,\Sigma_{\kappa}^{\pm}\right],\partial_{\kappa}\Sigma_{\kappa}^{\pm}\right]
  =2\left(1+\kappa^{2}\right)^{-\frac{3}{2}}
  \begin{bmatrix}\pm U^{\ast}\left[\Lambda,U\right] & 0\\
    0 & \mp U\left[\Lambda,U^{\ast}\right]
  \end{bmatrix}\,.
\end{eqnarray*}
When we now integrate over $\kappa\in\mathbb{R}$, we obtain, using
\[
  \int_{\kappa\in\mathbb{R}}2\left(1+\kappa^{2}\right)^{-\frac{3}{2}}\dif{\kappa}=4\,,
\]
that 
\eql{
  \mathcal{N}  =  \pm\frac{-1}{2}\tr\left(U^{\ast}\left[\Lambda,U\right]-U\left[\Lambda,U^{\ast}\right]\right)= \mp\tr\left(U^{\ast}\left[\Lambda,U\right]\right)= \pm\findex\left(\mathbb{\Lambda}U\right)\,.
}
We learn that 
\eql{
  \findex\left(\mathbb{\Lambda}S\right) =\pm\Chern\left(H\pm\kappa\sigma_{3}\right)
  =  \frac{1}{2}\left[\Chern\left(H+\kappa\sigma_{3}\right)-\Chern\left(H-\kappa\sigma_{3}\right)\right]\,.
}

\subsection{The bulk-edge correspondence}

Using the bulk-edge correspondence
we know that the difference of the Chern numbers between the left
and the right of $L_\kappa^\pm$ should yield the edge index of the interface domain-wall operator $L_{\kappa}$. 
Since we have the special case of $L_{\kappa}$ a fiber operator in
$\kappa$, the edge index is the spectral flow \cite{Graf07} of $\RR\ni\kappa\mapsto L_\kappa$, i.e.,
\[
  \ii\tr_{x,\kappa}\left(\RR\ni\kappa\mapsto g'\left(L_{\kappa}\right)\partial_{\kappa}L_{\kappa}\right)\,.
\]
We have thus established
\eq{\findex\br{\bbLambda S }  = \frac12\Sf\br{\RR\ni\kappa\mapsto L_\kappa}\,.}

Finally one may use symmetry to go from the flow on $\RR$ to the flow on $[0,\infty)$: first note that \eq{
  \sigma_3 L_{\kappa } \sigma_3 = -H +\kappa \sgn(X) \sigma_3 = -\br{H - \kappa \sgn(X) \sigma_3}= -L_{-\kappa}.
} But the spectral flow is unitarily invariant and odd under overall sign, we find 
\eq{
  \frac12\Sf\br{\RR\ni\kappa\mapsto L_\kappa} = \Sf\br{[0,\infty)\ni\kappa\mapsto L_\kappa}
} and hence the proof of \Cref{lem:BEC proof of index theorem} is now complete.

\appendix

\section{Locality estimates}

\subsection{Commutator estimates}
\begin{lem}\label{eq:commutator with position operator is bounded}
  Let $A\in\calB(\ell^2(\ZZ^d)\otimes\CC^N)$ be exponentially local as in \eq{
    \norm{A_{xy}} \leq C \exp\br{-\mu\norm{x-y}}\qquad(x,y\in\ZZ^d)\,.
  } Then for all $j=1,\cdots,d$,
  \eq{
    \norm{\left[A,X_j\right]} \leq C\frac{\coth\br{\frac{\mu}{2\sqrt{d}}}^{d-1}}{\cosh\br{\frac{\mu}{\sqrt{d}}}-1} < \infty\,.
  } 
\end{lem}
\begin{proof}
  We use the Holmgren-Schur bound \eql{\label{eq:Holmgren-Schur}
    \norm{A} \leq \sqrt{\sup_i\sum_{j}\abs{A_{ij}}}\sqrt{\sup_j\sum_{i}\abs{A_{ij}}}
  } to estimate
  \eq{
    \sup_{x\in\ZZ^d}\sum_{y\in\ZZ^d}\norm{\left[A,X_j\right]_{xy}} &\leq \sup_{x\in\ZZ^d}\sum_{y\in\ZZ^d}\abs{x_j-y_j}C\exp\br{-\mu\norm{x-y}} \\
    &\leq \sup_{x\in\ZZ^d}\sum_{y\in\ZZ^d}\abs{x_j-y_j}C\exp\br{-\frac{\mu}{\sqrt{d}}\norm{x-y}_1}\\
    &= C\frac{\coth\br{\frac{\mu}{2\sqrt{d}}}^{d-1}}{\cosh\br{\frac{\mu}{\sqrt{d}}}-1}\,.
  } 
  
\end{proof}

\begin{lem}\label{lem:commutator bounds with fl}
  Let $A\in\calB(\ell^2(\ZZ)\otimes\CC^N)$ be exponentially local as in \eq{
    \norm{A_{xy}} \leq C \exp\br{-\mu\abs{x-y}}\qquad(x,y\in\ZZ)\,.
  } Let $f_\ell$ be defined as in \eqref{eq:fell_def}, i.e., \eql{
    \ZZ \ni x \mapsto \begin{cases}
      1 & x=0,\\
      x & 1\leq|x|<\ell, \\
      \ell \sgn(x) & |x|\geq\ell,
    \end{cases}
  } with
  \eql{
    \sgn(x) \equiv \begin{cases}
      1 & x\geq0\\
      -1 & x <0
    \end{cases}\qquad(x\in\RR)\,.
  } Then
  \eql{\label{eq:operator norm commutator bound}
    \sup_{\ell}\norm{\left[A,f_\ell(X)\right]} \leq \frac{C}{\cosh\br{\mu}-1}
  } and 
  \eql{\label{eq:trace norm commutator bound}
    \norm{\left[A,f_\ell(X)\right]}_1 \leq \ell\norm{\left[A,\sgn(X)\right]}_1 + 2N\norm{A}\ell\br{2\ell+1}+4N\norm{A}\,.
  }
\end{lem}
\begin{proof}
  We estimate again via Holmgren \cref{eq:Holmgren-Schur} \eq{
    \sup_{x\in\ZZ}\sum_{y\in\ZZ} \abs{f_\ell(x)-f_\ell(y)}C\exp\br{-\mu\abs{x-y}}\,.
  }
  Now note that $\abs{f_\ell(x)-f_\ell(y)}\leq\abs{x-y}$ uniformly in $\ell$, so using the proof above we find an upper bound on the above expression as 
  \eq{
    \frac{C}{\cosh\br{\mu}-1}\,.
  } 

  To get the trace-class bound, we see that \eq{
    f_\ell(X)=\ell\sgn(X)+\chi_{[-\ell,\ell]}(X)\br{X-\ell\sgn(X)}+P_0
  }
  where $P_0:=\br{\delta_0\otimes\delta_0^\ast}\otimes\Id_{N\times N}$. We know that $\norm{\left[A,\sgn(X)\right]}_1<\infty$ \cite[Lemma 2 (b)]{Graf_Shapiro_2018_1D_Chiral_BEC} and 
  \eq{
    \norm{\left[A,\chi_{[-\ell,\ell]}(X)\br{X-\ell\sgn(X)}\right]}_1 &\leq 2N\norm{A}\ell\br{2\ell+1}\,.
  }
\end{proof}

\begin{lem}\label{lem:commutator bounds in two dimensions}
  Let $A\in\calB(\ell^2(\ZZ^2)\otimes\CC^N)$ be exponentially local as in \eq{
    \norm{A_{xy}} \leq C \exp\br{-\mu\norm{x-y}}\qquad(x,y\in\ZZ^2)\,.
  } Let $f_\ell(X)$ be defined via \eql{
    f_\ell(X)=\ell L+\chi_{[-\ell,\ell]^2}(X)\br{X-\ell L}+P_0
  } with $L:=\ee^{\ii\arg\br{X_1+\ii X_2}}$.
  Then
  \eq{
    \sup_{\ell}\norm{\left[A,f_\ell(X)\right]} \leq \br{\frac{C}{\cosh\br{\mu}-1}}^2
  } and 
  \eq{
    \norm{\left[A,f_\ell(X)\right]}_3 \leq
    \ell\norm{\left[A,L\right]}_3 +
    2N\norm{A}\ell\br{2\ell+1}^2+4N\norm{A}\,.}

\end{lem}

The proof is identical to the one above and is thus omitted, here we
use that $\sup_\ell\abs{f_\ell(x)-f_\ell(y)}\leq\norm{x-y}$ and that $\norm{\left[A,L\right]}_3<\infty$ (see e.g. \cite[Lemma A.1]{BSS23}).

\section{Estimates on gaps}

For any (not necessarily self-adjoint) invertible operator $A$, $$\gap(A)\equiv\norm{A^{-1}}^{-1}=\dist(0,\spec(|A|))\,.$$  If $A$ is self-adjoint, this is equivalent to $\gap(A) = \dist(0,\spec(A))$.  For non-invertible operators $A$ we set $\gap(A) = 0$ and note that $\gap(A)>0$ if and only if $A$ is invertible.

\begin{lem}[Basic estimates on the gap]
  \label{lem:basic_gap}
  The spectral gap satisfies the following estimates:
  \begin{enumerate}
  \item If $A$ is invertible then
    \begin{equation}
      \label{eq:gap_squared}
      \gap(A) = \sqrt{\gap(A^\ast A)}.  
    \end{equation}
  \item If $A,B$ are both invertible,
    \begin{equation}
      \label{eq:gapAB}
      \gap(A) \gap(B)
      \leq \gap(AB) 
      \leq \gap(A) \|B\|.
    \end{equation}
  \item
    \begin{equation}
      \label{eq:gapAplusB}
      \gap(A) - \|B\| \leq \gap(A+B) \leq 
      \gap(A) + \|B\|.
    \end{equation}
  \item \label{item:gap_additive}
    If $A$ and $B$ are non-negative definite,
    \begin{equation}
      \label{eq:gapAplusBpositive}
      \gap(A+B) \geq \gap(A) + \gap(B).
    \end{equation}
  \item
    \begin{equation}
      \label{eq:gapDirectSum}
      \gap(A\oplus D) = \min\{\gap(A), \gap(D)\}.
    \end{equation}
  \end{enumerate}
\end{lem}

\begin{proof}
  One has
  \begin{equation}
    \| (AB)^{-1} \|
    = \| B^{-1} A^{-1} \|
    \leq \| B^{-1}\| \|A^{-1}\|
    = \frac1{\gap(A) \gap(B)},
  \end{equation}
  establishing the left side of \eqref{eq:gapAB}.  The right side
  follows by using the left side to note that
  \begin{equation*}
    \gap(AB) \gap(B^{-1}) \leq \gap(ABB^{-1}) = \gap(A).
  \end{equation*}
  
  To show the left side of \eqref{eq:gapAplusB}, we note that the case $\gap(A) \leq \|B\|$ is vacuous, therefore we assume $\gap(A) > \|B\|$, which implies that $A$ is invertible and $\|A^{-1}B\|<1$.
  The estimate (valid for $\|C\|<1$)
  \begin{equation}
    \label{eq:resolvent_est}
    \left\| (\Id + C)^{-1} \right\|
    \leq \frac1{1-\|C\|}
  \end{equation}
  gives
  \begin{equation}
    \gap\left(I + A^{-1} B\right) 
    \geq 1 - \|A^{-1} B \|
    \geq 1 - \|A^{-1}\| \|B\|
    = 1 - \frac{\|B\|}{\gap(A)}.
  \end{equation}
  The left side of \eqref{eq:gapAplusB} now follows by applying \eqref{eq:gapAB} to $A+B = A(\Id+A^{-1}B)$.  Exchanging $A \mapsto A+B$, $B\mapsto -B$ in the left side of \eqref{eq:gapAplusB} yields the right side of \eqref{eq:gapAplusB}.
  
\end{proof}

\begin{lem}[Gaps of block operators]
  Let a self-adjoint $H$ on $\calH \oplus \calH^\perp$ be given by
  \eql{
    \label{eq:block_form}
    H = 
    \begin{bmatrix}
      A & B \\
      B^* & D
    \end{bmatrix}\,.
  } Then
  \begin{enumerate}
  \item We have a lower bound on the gap of the full operator in terms of the gaps of the diagonal blocks:
    \begin{equation}
      \label{eq:block_gap}
      \gap(H) \geq \min\{\gap(A), \gap(D)\} - \|B\|.
    \end{equation}
  \item We have an upper bound on the gap of the full operator in terms of the Schur gap:
    \begin{equation}
      \label{eq:UpperSchur}
      \gap(H) \leq
      \min\br{\Set{\gap(D),\gap\left(A - B D^{-1} B^*\right)}}
      \left(1+ \left\|B D^{-1}\right\|\right)^2.
    \end{equation}
  \item If a self-adjoint $H$ given by \eqref{eq:block_form} satisfies $\gap(A) \leq \gap(D)$,
    \begin{equation}
      \label{eq:block_gap_D}
      \gap(H)
      \geq \gap(A) - \sqrt{\frac{\gap(A)}{\gap(D)}} \| B \| 
    \end{equation}
  \end{enumerate}
\end{lem}

We remark that the estimate \eqref{eq:block_gap} is useful if $B$ is small, while \eqref{eq:UpperSchur} and \eqref{eq:block_gap_D} come into play when $D$ has large gap.

\begin{proof}
  To show \eqref{eq:block_gap}, we note that
  \begin{equation}
    \gap\left(
      \begin{bmatrix}
        A & 0 \\
        0 & D
      \end{bmatrix}
    \right)
    = \min\{\gap(A),\gap(D)\},
    \qquad
    \left\|
      \begin{bmatrix}
        0 & B \\
        B^* & 0
      \end{bmatrix}
    \right\| = \|B\|,
  \end{equation}
  and apply \eqref{eq:gapAplusB}.

  To show \eqref{eq:UpperSchur}, we write
  \begin{equation}
    H = 
    \begin{bmatrix}
      \Id & B D^{-1}\\
      0 & \Id
    \end{bmatrix}
    \begin{bmatrix}
      A-B D^{-1}B^* & 0 \\
      0 & D
    \end{bmatrix}
    \begin{bmatrix}
      \Id & 0\\
      D^{-1}B^* & \Id
    \end{bmatrix},
  \end{equation}
  and use the right estimate in \eqref{eq:gapAB}.  According to \eqref{eq:gapDirectSum}, the gap of the middle matrix is bounded by $\gap\left(A-B D^{-1}B^*\right)$.  The norms of the left and right matrices are bounded by $1+\left\|B D^{-1}\right\|$. 

  Finally, writing
  \begin{equation}
    H = 
    \begin{bmatrix}
      \Id & 0\\
      0 & \frac1\varepsilon \Id
    \end{bmatrix}
    \begin{bmatrix}
      A & \varepsilon B \\
      \varepsilon B^* & \varepsilon^2 D
    \end{bmatrix}
    \begin{bmatrix}
      \Id & 0\\
      0 & \frac1\varepsilon \Id
    \end{bmatrix},
  \end{equation}
  we estimate (assuming $0 < \varepsilon < 1$ and therefore the norm of the outside factors is $1$),
  \begin{equation}
    \gap(H) \geq \min\{\gap(A), \varepsilon^2\gap(D)\} - \varepsilon \|B\|.
  \end{equation}
  Setting
  \begin{equation}
    \varepsilon = \sqrt{\frac{\gap(A)}{\gap(D)}}
  \end{equation}
  yields \eqref{eq:block_gap_D}.
\end{proof}

\begin{lem}
  \label{lem:invertible_cut}
  For any $s\in[0,1]$, let a self-adjoint operator $H_s$ on $\calH \oplus \calH^\perp$ be given by
  \begin{equation}
    \label{eq:cut_flow}
    H_s = 
    \begin{bmatrix}
      A & s B \\
      s B^* & D
    \end{bmatrix}.
  \end{equation}
  If 
  \eql{
    \label{eq:D_large}   \gap(D)>\norm{B}\max\br{\Set{1,\frac{\norm{B}}{\frac14\gap(H_1)}}}
  }
  then $H_s$ is invertible for all $s\in[0,1]$.
\end{lem}

\begin{proof}
  Conditions \eqref{eq:D_large} imply that $D$ is invertible.
  We can then write
  \begin{equation}
    \label{eq:Schur_expansion}
    H_s = 
    \begin{bmatrix}
      \Id & sB D^{-1}\\
      0 & \Id
    \end{bmatrix}
    \begin{bmatrix}
      A-s^2B D^{-1}B^* & 0 \\
      0 & D
    \end{bmatrix}
    \begin{bmatrix}
      \Id & 0\\
      sD^{-1}B^* & \Id
    \end{bmatrix}.
  \end{equation}
  The matrices on the left and on the right are invertible because from \eqref{eq:D_large} we get that $\|sBD^{-1}\|<1$.

  From inequality \eqref{eq:UpperSchur} we get
  \begin{equation*}
    \gap\left(A - B D^{-1} B^*\right)
    \geq \frac{\gap(H_1)}{
      \left(1+ \left\|B D^{-1}\right\|\right)^2}
    \geq \frac14 \gap(H_1),
  \end{equation*}
  and therefore by \Cref{eq:gapAplusB},
  \begin{align*}
    \gap\left(A - s^2 B D^{-1} B^*\right)
    &\geq 
      \gap\left(A - B D^{-1} B^*\right) 
      - \big(1-s^2) \left\|B D^{-1} B^*\right\| \\
    &\geq \frac14 \gap(H_1) - \frac{\|B\|^2}{\gap(D)},
  \end{align*}
  which is positive by \eqref{eq:D_large}.  We conclude that the middle matrix in \eqref{eq:Schur_expansion} is invertible for all $s\in[0,1]$, and so is $H_s$.
\end{proof}

\begin{lem}
  \label{lem:block_traceless}
  Let a self-adjoint operator $H_s$ on $\calH \oplus
  \calH^\perp$ be given by
  \eq{
    L:=\begin{bmatrix}
      A & B^\ast \\ B & -A
    \end{bmatrix}}
  where $A=A^\ast$ and $B$ is general.  Then
  \begin{equation}
    \label{eq:gap_traceless}
    \gap(L) \geq
    \sqrt{\gap(A)^2+\gap(B)^2-\norm{[A,B]}}.
  \end{equation}
\end{lem}

\begin{proof}
  We have \eq{L^2=
    \begin{bmatrix}
      A & B^\ast \\ B & -A
    \end{bmatrix}^2 = \begin{bmatrix}
      A^2 + \abs{B}^2 & 0 \\ 0 & A^2 + \abs{B^\ast}^2
    \end{bmatrix} + \begin{bmatrix}
      0 & [A,B^\ast] \\ 
      -[A,B] & 0
    \end{bmatrix}\,.
  }
  We use \Cref{lem:basic_gap} to estimate
  \eq{
    \gap\br{L^2} 
    &\geq \gap\br{\begin{bmatrix}
          A^2 + \abs{B}^2 & 0 \\ 0 & A^2 + \abs{B^\ast}^2
        \end{bmatrix}}-\norm{\begin{bmatrix}
          0 & [A,B^\ast] \\ 
          -[A,B] & 0
        \end{bmatrix}} \\
    &\geq \gap(A)^2+\gap(B)^2-\norm{[A,B]}.
  }
  The conclusion for $\gap(L)$ now follows by \eqref{eq:gap_squared}.
\end{proof}


\section{A short review of the spectral flow}\label{sec:spectral flow}
We give a brief review of the properties of the spectral flow we need; more details may be found in e.g. the textbook \cite{DollSchulzBaldes2023}.

\begin{defn}\label{def:spectral flow of a family}
  Let $\calH$ be a separable Hilbert space. Let $[0,1]\ni t \mapsto A_t$ be an analytic family of self-adjoint Fredholm operators on $\calH$ such that $A_0,A_1$ are moreover invertible. We call such a family \emph{a spectral-flow family}. Let $\Delta\subseteq\RR$ be the essential gap of the family, i.e., \eq{
    \sigma_{\mathrm{ess}}(A_t) \cap \Delta = \varnothing, \qquad t\in[0,1]\,.
  } We assume further that $\sigma(A_i)\cap\Delta=\varnothing$ for $i=0,1$ and that $\Delta$ contains an open interval about zero; let $g\in C^\infty(\RR\to[0,1])$ be so that $g=1$ above $\Delta$, $g=0$ below $\Delta$ and $\supp(g')\subseteq\Delta$. Then \eql{
    \Sf(A) := \int_{t\in[0,1]}\tr_\calH\br{g'(A_t)\partial_t A_t}\dif{t}\,.
  }
\end{defn}

Informally we interpret
$\tr\br{g'(A_t)\partial_t A_t} = \partial_t \tr(g(A_t))$ and apply the
fundamental theorem of calculus to get \eq{ \Sf(A) = \tr\br{g(A_1) -
    g(A_0)}, } which counts the number of eigenvalues which have
transitioned from below to above the gap as we interpolate between $A_0$ and
$A_1$. We refer the reader to \cite{Phi_cmb96} or
\cite[Sec.~4.1]{DollSchulzBaldes2023}) for the extension of the
spectral flow to norm-continuous families.  We also remark that the
interval $[0,1]$ can be replaced with any other interval.

In this work we use the following fundamental properties of the
spectral flow.

\begin{thm}
  \label{thm:spec_flow_basics}
  For any spectral flow family $t\mapsto A_t$, the spectral flow $\Sf(A_\cdot)$ is:
  \begin{enumerate}
  \item Well-defined (i.e., it is independent of the choices of $\Delta$ and $g$ above).
  \item Integer-valued: \eq{
      \Sf(A) \in \ZZ\,.
    }
  \item Additive \cite[Thm.~4.2.1(v)]{DollSchulzBaldes2023}: If $A,B$ are
    two spectral flow families, then \eq{ \Sf(A\oplus B) =
      \Sf(A)+\Sf(B) \,.  }
  \item Concatenation-additive
    \cite[Thm.~4.2.1(ii)]{DollSchulzBaldes2023}: If
    $[0,1] \ni t \mapsto A_t$ is norm-continuous with $A_0$,
    $A_{0.5}$ and $A_1$ invertible, then
    \begin{equation}
      \label{eq:concat}
      \Sf\big([0,1] \ni t \mapsto A_t\big)
      = \Sf\big([0,0.5] \ni t \mapsto A_t\big)
      + \Sf\big([0.5,1] \ni t \mapsto A_t\big).
    \end{equation}
  \item Homotopy invariant \cite[Thm.~4.2.2]{DollSchulzBaldes2023}: Let
    $[0,1]\times[0,1] \ni (t,s) \mapsto h(t,s)$ be a norm-continuous
    family of bounded Fredholm operators.  Consider two continuous paths
    $\gamma_1, \gamma_2 :[0,1] \to [0,1]^2$ with
    $\gamma_1(0)=\gamma_2(0)$ and $\gamma_1(1)=\gamma_2(1)$.  If
    $h\big(\gamma_1(0)\big)$ and $h\big(\gamma_1(1)\big)$ are
    invertible, then
    \eq{\Sf\Big([0,1]\ni t \mapsto h\big(\gamma_1(t)\big)\Big)
      = \Sf\Big([0,1]\ni t \mapsto h\big(\gamma_2(t)\big)\Big).}
  \end{enumerate}
\end{thm}

\bigskip
\bigskip
\noindent\textbf{Acknowledgments.} 
We are  indebted to Shinsei Ryu and Tom Stoiber for useful
discussions.  GB was partially supported by the NSF grant
DMS-2510345.  JS was partially suppored by the NSF grant DMS-2510207.

\begingroup
\let\itshape\upshape
\printbibliography
\endgroup
\end{document}